\documentclass[11pt]{article}
\usepackage{amsmath,amssymb,amsfonts}

\usepackage{latexsym,nicefrac,bbm}
\usepackage{xspace}
\usepackage{color,fancybox,graphicx,fullpage,soul}
\usepackage{hyperref}
\usepackage{pdfsync}
\usepackage{hhline}

\newcommand{\defeq}{\stackrel{\textup{def}}{=}}
\newcommand{\nfrac}{\nicefrac}

\newcommand{\tO}{\widetilde{O}}

\usepackage{enumitem}
\expandafter\def\expandafter\normalsize\expandafter{
\setlength{\abovedisplayskip}{3pt}
\setlength{\belowdisplayskip}{3.5pt}
\setlength{\abovedisplayshortskip}{3pt}
\setlength{\belowdisplayshortskip}{3.5pt}}
\newcommand{\cM}{\mathcal{M}}
\newcommand{\cC}{\mathcal{C}}
\newcommand{\cB}{\mathcal{B}}

\newcommand{\Pb}{\mathbb{P}}

\newcommand{\ccount}{\textsc{Count}}
\newcommand{\csample}{\textsc{Sample}}
\newcommand{\bcount}{\textsc{BCount}}
\newcommand{\bsample}{\textsc{BSample}}
\newcommand{\ecount}{\textsc{ECount}}

\newcommand{\eps}{\varepsilon}

\newenvironment{proof}{\noindent{\bf Proof:}\hspace*{1em}}{\qed\bigskip}
\newcommand{\qed}{\hfill\ensuremath{\square}}

\def\showauthornotes{0}

\def\showdraftbox{0}

\newtheorem{theorem}{Theorem}[section]

\newtheorem{definition}{Definition}[section]
\newtheorem{lemma}[theorem]{Lemma}
\newtheorem{remark}[theorem]{Remark}

\newtheorem{corollary}{Corollary}[section]

\newtheorem{fact}[theorem]{Fact}

\def\FullBox{\hbox{\vrule width 6pt height 6pt depth 0pt}}

\def\qed{\ifmmode\qquad\FullBox\else{\unskip\nobreak\hfil
\penalty50\hskip1em\null\nobreak\hfil\FullBox
\parfillskip=0pt\finalhyphendemerits=0\endgraf}\fi}

\def\qedsketch{\ifmmode\Box\else{\unskip\nobreak\hfil
\penalty50\hskip1em\null\nobreak\hfil$\Box$
\parfillskip=0pt\finalhyphendemerits=0\endgraf}\fi}

\newenvironment{proofof}[1]{\begin{trivlist} \item {\bf Proof
#1:~~}}
  {\qed\end{trivlist}}

\newcommand\Z{\mathbb Z}
\newcommand\N{\mathbb N}
\newcommand\R{\mathbb R}

\newcommand{\marginlabel}[1]%
{\mbox{}\marginpar{\it{\raggedleft\hspace{0pt}#1}}}
\newcommand{\poly}{\mathrm{poly}}
\definecolor{Mygray}{gray}{0.8}

 \ifcsname ifcommentflag\endcsname\else
  \expandafter\let\csname ifcommentflag\expandafter\endcsname
                  \csname iffalse\endcsname
\fi

\ifnum\showauthornotes=1
\newcommand{\Authornote}[2]{{\sf\small\color{red}{[#1: #2]}}}
\newcommand{\Authoredit}[2]{{\sf\small\color{red}{[#1]}\color{blue}{#2}}}
\newcommand{\Authorfnote}[2]{\footnote{\color{red}{#1: #2}}}
\newcommand{\Authorfixme}[1]{\Authornote{#1}{\textbf{??}}}
\newcommand{\Authormarginmark}[1]{\marginpar{\textcolor{red}{\fbox{%\Large
#1:!}}}}
\else
\newcommand{\Authornote}[2]{}
\newcommand{\Authoredit}[2]{}
\newcommand{\Authorcomment}[2]{}
\newcommand{\Authorfnote}[2]{}
\newcommand{\Authorfixme}[1]{}
\newcommand{\Authormarginmark}[1]{}
\fi

\newcommand{\inparen}[1]{\left(#1\right)}             %\inparen{x+y}  is (x+y)
\newcommand{\inbraces}[1]{\left\{#1\right\}}           %\inbrace{x+y}  is {x+y}
\newcommand{\norm}[1]{\ensuremath{\left\lVert #1 \right\rVert}}

\newcommand{\diag}[1]{{\sf Diag}\left({#1}\right)}

 {
	\begin{enumerate}}{\end{enumerate}}

\newcounter{lecnum}

\newlength{\tpush}
\ifnum\showdraftbox=1

\else

\fi

\title{On the Complexity of Constrained Determinantal Point Processes\footnote{This paper subsumes the following prior unpublished works: arxiv 1607.01551 and 1608.00554.}}

\usepackage{authblk}

\author[1]{L. Elisa Celis\thanks{elisa.celis@epfl.ch}}
\author[2]{Amit Deshpande\thanks{amitdesh@microsoft.com}}
\author[2]{Tarun Kathuria\thanks{t-takat@microsoft.com}}
\author[1]{Damian Straszak\thanks{damian.straszak@epfl.ch}}
\author[1]{Nisheeth K. Vishnoi\thanks{nisheeth.vishnoi@epfl.ch}}
\affil[1]{\'{E}cole Polytechnique F\'{e}d\'{e}rale de Lausanne (EPFL), Switzerland}
\affil[2]{Microsoft Research, India}
\begin{document}
\maketitle
\begin{abstract}

Determinantal Point Processes (DPPs) are probabilistic models that arise in quantum physics and random matrix theory and have recently found numerous applications in theoretical computer science and machine learning.
DPPs define probability distributions over subsets of a given ground set, they exhibit interesting properties such as negative correlation, and, unlike other models of negative correlation such as Markov random fields, have efficient algorithms for sampling.
When applied to kernel methods in machine learning, DPPs favor subsets of the given data with more diverse features.
However, many real-world applications require efficient algorithms to sample from DPPs with additional constraints on the sampled subset, e.g., partition or matroid constraints that are important from the viewpoint of ensuring priors, resource or fairness constraints on the sampled subset.
Whether one can efficiently sample from DPPs in such constrained settings is an important problem that was first raised in an influential survey of DPPs for machine learning by \cite{KuleszaTaskar12} and studied in some recent works in the machine learning literature.
The main contribution of our paper is the first correct resolution of the complexity of sampling from DPPs with constraints.
On the one hand, we give exact efficient algorithms for sampling from constrained DPPs when the description of the constraints is  in unary; this includes special cases of practical importance such as a small number of partition, knapsack or budget constraints.
On the other hand, we prove that when the constraints are specified in binary, this problem is $\#${\bf P}-hard via a reduction from the problem of computing mixed discriminants implying that it may be unlikely that there is an FPRAS.
Technically, our algorithmic result benefits from viewing the constrained sampling  problem via the lens of  polynomials and we obtain our complexity results by providing an equivalence between computing mixed discriminants and sampling from partition constrained DPPs. 
As a consequence,  we obtain a few  corollaries of independent interest: 1) An  algorithm to count, sample (and, hence, optimize) over the base polytope of regular matroids when there are additional (succinct) budget constraints and, 
2) An algorithm to  evaluate and compute the { mixed characteristic polynomials}, that played a central role in the resolution of the Kadison-Singer problem, for certain special cases.
\end{abstract}

\section{Introduction}
Algorithms for  sampling from a discrete set of objects are sought after in various disciplines of  computer science, optimization, mathematics and physics due to their far reaching applications.
For instance, sampling from the Gibbs distribution was one  of the original optimization methods (see, e.g., \cite{Hajek88}) and sampling from dependent distributions is often used in the design of  approximation algorithms (see, e.g., \cite{BertsimasV04,ChekuriVZ10,Harvey2014}).
In machine learning, algorithms for sampling from discrete probability distributions are sought after in various summarization, inference and learning tasks \cite{JordanWainwright,MezardMontanari,KuleszaTaskar12}.
Here, a particular class of probability distributions that has received much attention are the Determinantal Point Processes (DPP).
In the discrete setting, a DPP is a distribution over subsets of a finite data set $[m] \defeq \{1,2,\ldots, m\}$.
Here, a data point  $i$ is associated to a  feature vector $v_i \in \mathbb{R}^d$, and an $m \times m$ positive semidefinite (PSD) kernel $L$ gives the dot product of the feature vectors of any two data points as a measure of their pairwise similarity.
Determinants, then, provide a natural measure of the diversity of a subset of data points, often backed by a physical intuition based on volume or entropy.
A  DPP is thus defined with respect to the kernel $L$ such that for all $S\subseteq [m]$ we have $\Pb(S)\propto \det(L_{S,S})$, where $L_{S,S}$ is the principal minor of $L$ corresponding to rows and columns from $S$.\footnote{We treat DPPs via {\it $L$-ensembles}, while commonly they are defined using   {\it kernel matrices}, for practical purposes these two definitions are equivalent.}
The quantity $\det(L_{S,S})$ can be interpreted as the squared volume of the $|S|$-dimensional parallelepiped spanned by the vectors $\{v_i: i\in S\}$ and, intuitively, the larger the volume, the more {\em diverse} the set of vectors.
Hence such distributions tend to prefer most diverse or {\em informative}  subsets of data points.
Mathematically, the fact that the probabilities are derived from determinants allows one to deduce elegant and non-trivial properties of such distributions,  such as negative correlation and concentration of measure. %
Efficient polynomial time algorithms for sampling from DPPs (see \cite{HKPV05,DR10}) is what sets them apart from the other probabilistic models of negative correlation such as Markov random fields.
As a consequence, sampling from DPPs has been  successfully applied to a number of problems, such as document summarization, sensor placement and recommendation systems \cite{Lin11,Krause08,Zhou09,Zhai03,Yue08}.

Given the wide applicability of DPPs, a natural question is whether they can be generalized to incorporate priors, budget or fairness constraints, or other natural combinatorial constraints. % 
In other words, given an $m \times m$ kernel $L$ and a family  $\cC \subseteq 2^{[m]}$ that represents  constraints on the subsets, can we efficiently sample from the DPP distribution supported only on $\cC$; that is, $\Pb(S) \propto \det(L_{S,S})$ for $S \in \cC$, and $\Pb(S) = 0$ otherwise. Here are two important special cases.
\begin{itemize}
\item \emph{Fairness (or Partition) constraints:} Consider the setting where $[m]$ is a collection of data points and each point is associated with a sensitive attribute such as gender. Then $\cC$ is the family of attribute-unbiased subsets of $[m]$ -- e.g., those subsets that  contain an equal number of male and female points.
Thus, the corresponding $\cC$-constrained DPP outputs a diverse set of points while maintaining fairness with respect to the sensitive attribute; see \cite{CDKV16} for this and other applications of constrained DPPs to eliminating algorithmic bias.
\item \emph{Budget constraints:} In data subset selection or active learning, when there is a cost $c_i \in \mathbb{Z}$ associated with each data point, it is natural to ask for a diverse training sample $S$ from a corresponding DPP such that its cost $\sum_{i \in S}c_i$ is bounded from above by $C \in \mathbb{Z}$. See also \cite{WeiIyerBilmes15} for a related optimization variant.
\end{itemize}
In their survey, \cite{KuleszaTaskar12} posed the open question of efficiently sampling from DPPs with additional combinatorial constraints on the support of the distribution.
Sampling from constrained DPPs is algorithmically non-trivial, as many natural heuristics fail. The probability mass on the constrained family of subsets can be arbitrarily small, hence, ruling out a rejection sampling approach.
For partition constraints, a natural heuristic is to sample from independent smaller-sized DPPs, each defined over a different part. 
However, such a product distribution would select two (potentially very similar) items from two different parts independently, whereas in a constrained DPP distribution they must be negatively correlated. 
Unlike DPPs and the special case of cardinality-constrained $k$-DPPs (in which $\cC$ is the family of \emph{all} subsets of size $k$ -- see Section~\ref{sec:relatedwork}), it is not clear that there is a clean expression for the partition function or the marginals of a constrained DPP. 
Another approach to approximately sample from constrained DPPs is via Markov Chain Monte Carlo (MCMC) methods as in the recent work of \cite{LJS16}. 
This approach can be shown to be efficient when the underlying Markov chain is connected and the DPP kernel is close to a diagonal matrix (or nearly-log-linear; see Theorem 4 of \cite{LJS16}). 
However, the above conditions do not hold for sampling partition-constrained subsets -- even with constant number of parts -- from most DPP kernels.
Thus, while the problem of sampling from constrained DPPs has attracted attention, its complexity has remained open.

The main contribution of our paper is the first correct resolution of the problem of sampling from constrained DPPs.
Our  results give a dichotomy for the complexity of this problem: On the one hand, we give exact  algorithms which are polynomial time when the description of $\cC$ (in terms of the costs and budgets) is  in {\em unary}; this includes special cases of practical importance such as the fairness,  partition or budget constraints mentioned above. 
On the other hand, we prove that in general this problem is $\#${\bf P}-hard  when the constraints of $\cC$ are specified in binary.
Our algorithmic results go beyond the MCMC methods and include special cases of practical importance such as (constantly-many) partition or fairness constraints (studied, e.g., by \cite{CDKV16}) and a more general class of budget constraints and linear families defined in the following section.

Our algorithmic results benefit from viewing the probabilities arising in constrained DPPs as coefficients of certain multivariate polynomials.
This viewpoint also allows us to extend our result on constrained DPPs to derive important consequences of independent interest.
For instance, using the intimate connection between linear matroids and DPPs, we arrive at efficient algorithms to sample a basis of regular matroids when there are additional budget constraints -- significantly extending results of \cite{Eppstein95,BM97} for spanning trees.
To prove the hardness result, we present an \textit{equivalence} between the problem of \emph{computing} the {mixed discriminant} of a tuple of PSD matrices and that of \emph{sampling} from partition-constrained DPPs. 
Mixed discriminants (see Section \ref{sec:mixed_discr} for a definition) generalize the permanent,  arise in the  proof of the Kadison-Singer problem (\cite{MSS13}, see  \cite{Harvey13} for a survey on this topic) and are closely related to mixed volumes (see, e.g., \cite{Barvinok97}). 
However, unlike the result for permanent \cite{JSV04} and volume computation \cite{DyerFK91}, there is evidence  that the mixed discriminant problem may be much harder and may not admit an FPRAS; see \cite{Gurvits05}. 
Thus, in light of our equivalence between mixed discriminants and partition DPPs, it may be unlikely that we can even approximately sample from partition DPPs (with an arbitrary number of parts) efficiently.
Further, this connection  implies that important special cases of the mixed discriminant problem, for instance computing the higher order coefficients of the mixed-characteristic polynomial or evaluating the mixed characteristic polynomial of low rank matrices at a given point, can be solved efficiently, which may be of independent interest.

\subsection{Our Framework and Results}

The starting point of our work is the observation that if we let $\mu$ be the measure on subsets of $[m]$  corresponding to the kernel matrix $L$ (i.e., $\mu(S) \defeq \det(L_{S,S}))$, then given $L$, there is an efficient algorithm to evaluate the polynomial
$$g_\mu(x) \defeq \sum_{S\subseteq [m]} \mu(S) x^S$$
where $x^S$ denotes $\prod_{i\in S} x_i$ for any setting of its variables.
Indeed, consider the Cholesky decomposition of the kernel $L=VV^\top$.  Then, the polynomial $x \mapsto \det( V^\top X V +I)$ (where $X$ denotes the diagonal matrix with $x$ on the diagonal) is equal to $g_\mu(x)$ (see Fact~\ref{fact:gen_det}) and hence can be efficiently evaluated using Gaussian elimination for any input $x$.
We say that such a $\mu$ has an {\em efficient evaluation oracle} and, as it turns out,   this is the {\em only} property we need from DPPs and our results generalize to any measure $\mu$ for which we have such an evaluation oracle.
Before we explain our results, we formally introduce the sampling problem in this general framework.

\begin{definition}[Sampling]
Let $\mu : 2^{[m]} \to \R_{\geq 0}$ be a function assigning non-negative real values to subsets of $[m]$ and let $\cC \subseteq 2^{[m]}$ be any family of subsets of $[m]$.
We denote the (sampling) problem of selecting a set $S\in \cC$ with probability $p_S = \frac{\mu(S)}{\sum_{T\in \cC} \mu(T)}$ by $\csample[\mu, \cC].$
\end{definition}
\noindent
Building up on the  equivalence between sampling and counting \cite{JVV86}, we show that if one is given oracle access to the generating polynomial $g_\mu$ and
if $\mu$ is a nonnegative measure, the problem  $\csample[\mu, \cC]$ is  essentially equivalent to the following counting problem; see Theorem \ref{thm:equiv_count_sample} in Section \ref{sec:equiv}.
\begin{definition}[Counting]
Let $\mu : 2^{[m]} \to \R_{\geq 0}$ be a function assigning non-negative real values to subsets of $[m]$ and let $\cC \subseteq 2^{[m]}$ be any family of subsets of $[m]$.
We denote the (counting)  problem of computing the sum $\sum_{S \in \cC} \mu(S)$   by $\ccount[\mu, \cC]$.
\end{definition}

\noindent
In particular, a polynomial time algorithm for $\ccount[\mu, \cC]$ can be  translated into a polynomial time algorithm for $\csample[\mu, \cC]$.
Interestingly, this relation holds no matter what $\cC$ is; in particular, no specific assumptions on how the access to $\cC$ is provided are required.

Towards developing counting algorithms in our framework,  we focus on a class of families $\cC \subseteq 2^{[m]}$, which we call {\it Budget Constrained Families}, where a cost vector $c\in \Z^m$ and a budget value $C\in \Z$ are given, and the family consists of all sets $S\subseteq [m]$ of total cost $c(S) \defeq \sum_{i\in S} c_i$ at most $C$.
We call the counting and sampling problems for this special case $\bcount[\mu, c, C]$ and $\bsample[\mu, c, C]$ respectively.

Our key result is that  the  $\bcount$ problem (and hence also $\bsample$) is efficiently solvable whenever the costs are not too large in magnitude.

\begin{theorem}[Counting under Budget Constraints]\label{thm:main}
There is an algorithm, which given a function $\mu: 2^{[m]} \to \R$ (via oracle access to $g_\mu$), a cost vector $c\in \Z^m$ and a cost value $C\in \Z$ solves the $\bcount[\mu, c, C]$ problem in polynomial time with respect to $m$ and $\norm{c}_1$.
\end{theorem}
The proof of Theorem~\ref{thm:main} (see Section~\ref{sec:counting}) benefits from an interplay between probability measures and polynomials.
It reduces the counting problem to computing the coefficients of a certain univariate polynomial which, in turn, can be evaluated efficiently given access to the generating polynomial for $\mu$.
We can then employ interpolation in order to recover the required coefficients.

It is not hard to see that Theorem~\ref{thm:main} also implies the same result for families with a single {\it equality } constraint ($c(S)=C$) or for any constraint of the form $c(S) \in K$, where $K\subseteq \Z$ is given as input together with $c\in \Z^m$ and $C\in \Z$. Furthermore, our framework can be easily extended to the case of multiple (constant number of) such constraints.

As mentioned earlier, what makes DPPs attractive is that their generating polynomial, arising from a determinant, is  efficiently computable.
Using this fact, Theorem \ref{thm:main} and the equivalence between sampling and counting, we can deduce the following result.
\begin{corollary}\label{cor:dpps}
There is an algorithm, which given a PSD matrix $L\in \R^{m\times m}$, a cost vector $c\in \Z^m$ and a cost value $C\in \Z$ samples a set $S$ of cost $c(S)\leq C$ with probability proportional to $\det(L_{S,S})$. The running time of the algorithm is polynomial with respect to $m$ and $\norm{c}_1$.
\end{corollary}
From the above one can derive efficient sampling algorithms for several classes of constraint families $\cC$ which have {\it succinct descriptions}. Indeed, we establish counting and sampling algorithms for a general class of {\it linear families} of the form
\begin{equation}\label{eq:linear_fam}
\cC = \inbraces{S \subseteq [m]: c_1(S) \in K_1, c_2(S) \in K_2, \ldots, c_p(S) \in K_p}
\end{equation}
where $c_1, c_2, \ldots, c_p \in \Z^m$ and $K_1, \ldots, K_p \subseteq \Z$. We prove the following

\begin{corollary}\label{cor:dpps_gen}
There is an algorithm, which given a PSD matrix $L\in \R^{m\times m}$ and a description of a linear family $\cC$ as in~\eqref{eq:linear_fam}, samples a set $S\in \cC$ with probability proportional to $\det(L_{S,S})$. The running time of the algorithm is polynomial in $m$ and $\prod_{j=1}^p \inparen{\norm{c_j}_1+1}$.
\end{corollary}
One particular class of families for which the above yields polynomial time sampling algorithms are partition families (families of bases of partition matroids) over constantly many parts (see Corollary~\ref{cor:partition_sample}).
An important open problem that remains is to come up with even faster algorithms.

Another application of Theorem~\ref{thm:main}, which we present Section~\ref{sec:matroids}, is to combinatorial sampling and counting problems. More precisely, we note that the indicator measure of bases of regular matroids has an efficiently computable generating polynomial; hence, we can
solve their corresponding budgeted versions of counting and sampling problems.

One may ask if the dependence on $\|c\|_1$ in  Theorem~\ref{thm:main} can be improved.
We prove that the answer to this question is  no in a very strong sense.
To state our hardness result, we introduce $\ecount$ -- a natural variant of the $\bcount$ problem -- in which the sum is over subsets of cost equal to a given value $C$ instead of at most $C$ (such a problem is no harder than $\bcount$).
We provide an approximation preserving reduction showing that $\ecount[\mu, c, C]$ is at least as hard as computing {\em mixed discriminants} of tuples of positive semidefinite (PSD) matrices when $c$ and $C$ are given in binary, and can be exponentially large in magnitude.
Recall that for a tuple of  $m \times m$ PSD matrices $A_1,\ldots,A_m$, their mixed discriminant is the coefficient of the monomial $\prod_{i=1}^m x_i$ in the polynomial $\det (\sum_{i=1}^m x_i A_i)$.
\begin{theorem}[Hardness of Counting under Budget Constraints]\label{thm:hardness}
 $\bcount[\mu, c, C]$  is $\mathbf{\#P}-$hard. Moreover, when $\mu$ is a determinantal function, $\ecount[\mu,c, C]$ is at least as hard to approximate as mixed discriminants of tuples of PSD matrices.
\end{theorem}
To prove this result we show an \textit{equivalence} between the counting problem corresponding to partition-constrained DPPs (with a large, super-constant number of parts) and computing mixed discriminants.
Unlike permanents \cite{JSV04}, no efficient approximation scheme is known for estimating mixed discriminants and there is some evidence \cite{Gurvits05} that there may be none.
To further understand to what extent $g_\mu$ is the cause of computational hardness, in Section~\ref{sec:hardness_trees} (see Theorem~\ref{thm:trees}) we provide another hardness result; it considers a $\mu$ that is a $0/1$ indicator function for spanning trees in a graph (with efficiently computable $g_\mu$).
We prove that  $\ecount[\mu, c, C]$ is at least as hard to approximate as the number of perfect matchings in general (non-bipartite) graphs, which is another problem for which existence of an FPRAS is open.

Finally, this connection between partition-DPPs and mixed discriminants, along with our results to efficiently solve the counting problem for partition-DPPs with constantly many parts, gives us other applications of independent interest.
1) The ability to compute the top few coefficients of the {\em mixed characteristic polynomial} that arises in the proof of the Kadison-Singer problem; see Theorem \ref{thm:higherOrderThm}.
2) The ability to compute in polynomial time, the mixed characteristic polynomial \textit{exactly}, when the linear matrix subspace spanned by the input matrices has constant dimension; see Theorem \ref{thm:MCPreduction} and Corollary \ref{cor:rank}.

\subsection{Other Related Work}
\label{sec:relatedwork}

For sampling from $k$-DPPs there are exact polynomial time algorithms (see \cite{HKPV05,DR10,KuleszaTaskar12}). 
There is also recent work on faster approximate MCMC algorithms for sampling from various unconstrained discrete point processes (see \cite{RebeschiniK15} and the references therein),  and algorithms that are efficient for constrained DPPs under certain restrictions on the kernel and constraints (see \cite{LJS16} and the references therein). 
To the best of our knowledge, our result is the first efficient sampling algorithm that works for all kernels and for any constraint set with small description complexity.
On the practical side, diverse subset selection and DPPs arise in a variety of contexts such as structured prediction \cite{PrasadJB14}, recommender systems \cite{GartrellRecSys16} and active learning \cite{WeiIyerBilmes15}, where the study of DPPs with additional constraints is of importance. 
\section{Counting with Budget Constraints}
\label{sec:counting}
\begin{proofof}{of Theorem~\ref{thm:main}}
Let us first consider the case in which the cost vector $c$ is nonnegative, i.e.,  $c\in \N^m$. We introduce a new variable $z$ and consider the polynomial
$$h(z) \defeq g_\mu(z^{c_1}, z^{c_2}, \ldots, z^{c_m}).$$
Since $g_\mu(x_1, \ldots, x_m)=\sum_{S\subseteq [m]} \mu(S) \prod_{i\in S} x_i$, we have
$$h(z) = \sum_{S\subseteq [m]} \mu(S) \prod_{i \in S} z^{c_i} = \sum_{S\subseteq [m]} \mu(S) z^{c(S)}  = \sum_{0\leq d \leq \norm{c}_1} z^d  \sum_{S:~c(S)=d} \mu(S).$$
Hence, the coefficient of $z^d$ in $h(z)$ is equal to the sum of $\mu(S)$ over all sets $S$ such that $c(S)=d$.
In particular, the output is the sum of coefficients over $d\leq C$.

It remains to show how to compute the coefficients of $h$.
Note that we do not have direct access to $g_\mu$.
However, we can evaluate $g_\mu(x)$ at any input $x\in \R^m$, which in turn allows us to compute $h(z)$ for any input $z\in \R$.
Since $h(z)$ is a polynomial of degree at most $\norm{c}_1$, in order to recover the coefficients of $h$, it suffices to evaluate it at $\norm{c}_1+1$ inputs and perform interpolation.
When using FFT, the total running time becomes:
$$(\norm{c}_1+1)\cdot T_\mu +\tO(\norm{c}_1),$$
where $T_\mu$ is the running time of the evaluation oracle for $g_\mu$.

In order to deal with the case in which $c$ has negative entries, consider a modified version of $h$:
$$h(z) \defeq z^{\norm{c}_1} g_\mu(z^{c_1}, z^{c_2}, \ldots, z^{c_m}).$$
Clearly, $h(z)$ is a polynomial of degree at most $2\cdot \norm{c}_1$ whose  coefficients encode the desired output.
\end{proofof}

\noindent 
We also state a simple consequence of the above proof that is often convenient to work with.
\begin{corollary}\label{cor:equality}
There is an algorithm that, given a vector $c\in \Z^m$, a value $C\in \Z$ and oracle access to $g_\mu$ computes the sum $\sum_{S:~c(S)=C} \mu(S)$ in time polynomial with respect to $m$ and $\norm{c}_1$. 
\end{corollary}
In the above, note the equality $c(S)=C$ instead of $c(S)\leq C$ as in $\bcount$.

\section{Determinantal Point Processes}\label{sec:applications}

A Determinantal Point Process (DPP) is a probability distribution $\mu$ over subsets of $[m]$ defined with respect to a symmetric positive semidefinite matrix $L\in \R^{m\times m}$ by
$\mu(S) \propto \det(L_{S,S})$; i.e.,
$$\mu(S) \defeq \frac{\det(L_{S,S})}{\sum_{T\subseteq [m]} \det(L_{T,T})}.$$
We will often use a different matrix to represent the measure $\mu$; let $V\in \R^{m \times n}$ be a matrix, such that $L=VV^\top$ (the Cholesky decomposition of $L$).  Then, $\det(L_{S,S})=\det(V_S V_S^\top)$.

An important open problem related to DPPs is the sampling problem under additional combinatorial constraints imposed on the ground set $[m]$. We prove that these problems are polynomial time solvable for succinct budget constraints, as in Theorem~\ref{thm:main}. We start by establishing the fact that generating polynomials for determinantal distributions are efficiently computable.
\begin{fact}\label{fact:gen_det}
Let $L \in \R^{m \times m}$ be a PSD matrix with $L=VV^\top$ for some $V\in \R^{m\times n}$. If $\mu: 2^{[m]} \to \R_{\geq 0}$ is defined as $\mu(S) \defeq \det(L_{S,S})$ then $\det(V^\top XV +I) = \sum_{S\subseteq [m]} x^S \mu(S)$, where $X$ is the diagonal matrix of indeterminates $X=\diag{x_1, \ldots, x_m}$ and $I$ is the $n\times n$ identity matrix.
\end{fact}
\begin{proof}
We start by applying the Sylvester's determinant identity
$$\det(V^\top XV +I) =\det\inparen{\inparen{\sqrt{X} V}\inparen{\sqrt{X}V}^\top +I}.$$
It is well known that for a symmetric matrix $A\in \R^{m\times n}$ the coefficient of $t^k$  in the polynomial $\det(A+tI)$ is equal to $\sum_{|S|=n-k} \det(A_{S,S})$. Applying this result to $A=\inparen{\sqrt{X} V}\inparen{\sqrt{X}V}^\top$, we get
$$\det(A_{S,S}) =x^S \det(V_S V_S^\top)=x^S\det(L_{S,S}),$$
which concludes the proof.
\end{proof}
Now we are ready to deduce Corollary~\ref{cor:dpps}.
\begin{proofof}{of Corollary~\ref{cor:dpps}}
A polynomial time counting algorithm follows directly from Theorem~\ref{thm:main} and Fact~\ref{fact:gen_det}. To deduce sampling we apply the result on equivalence between sampling and counting Theorem~\ref{thm:equiv_count_sample}. In fact when applied to an exact counting algorithm we obtain an exact sampling procedure.
\end{proofof}
We move to the general result on sampling for linear families -- Corollary~\ref{cor:dpps_gen}. One can deduce it directly from Theorem~\ref{thm:main}, but this leads to a significantly suboptimal algorithm. Instead we take a different path and reprove Theorem~\ref{thm:main} in a slightly higher generality.

\begin{proofof}{of Corollary~\ref{cor:dpps_gen}}
We will show how to solve the counting problem -- sampling will then follow from Theorem~\ref{thm:equiv_count_sample}. Also, for simplicity we assume that all the entries in the cost vectors are nonnegative, this can be extended to the general setting as in the proof of Theorem~\ref{thm:main}.

Let $g$ be the generating polynomial of the determinantal function $\mu(S) = \det(L_{S,S})$, which is efficiently computable by Fact~\ref{fact:gen_det}. For notational clarity we will use superscripts to index constraints. For every constraint ``$c^{(j)}(S)\in K_i$'' ($j=1,2,\ldots, p$) introduce a new formal variable $y_j$. For every index $i\in [m]$ define the monomial:
\[
s_i = \prod_{j=1}^p y_j^{c^{(j)}_i}.
\]
The above encodes the cost of element $i$ with respect to all cost vectors $c^{(j)}$ for $j=1,2,\ldots,p$. Consider the polynomial $h(y_1, \ldots, y_p) = g(s_1, s_2, \ldots, s_m)$. It is not hard to see that the coefficient of a given monomial $\prod_{j=1}^p y_j^{d_j}$ in $h$ is simply the sum of $\mu(S)$ over all sets $S$ satisfying $c^{(1)}(S)=d_1, c^{(2)}(S)=d_2, \ldots, c^{(p)}(S)=d_p$. Hence the solution to our counting problem is simply the sum of certain coefficients of $h$. It remains to show how to recover all the coefficients efficiently.

Note that we can efficiently evaluate the polynomial $h$ at every input $(y_1, \ldots, y_p)\in \R^p$. One can then apply interpolation to recover all coefficients of $h$. The running time is polynomial in the total number of monomials in $h$, which can be bounded from above by
$\prod_{j=1}^p\inparen{\norm{c^{(j)}}_1+1}.$
\end{proofof}
We derive now one interesting application of Corollary~\ref{cor:dpps_gen} -- sampling from partition constrained DPPs. Let us first define {\it partition families} formally.
\begin{definition}
Let $[m]=P_1\cup P_2\cup \cdots \cup P_p$ be a partition of $[m]$ into disjoint,  nonempty sets and let $b_1, b_2, \ldots, b_p$ be integers such that $0 \leq b_i \leq |P_i|$. A family of sets of the form
\[
\cC = \{S\subseteq [m]: |S \cap P_j|=b_j,~\text{for every $j = 1, 2, \ldots, p$}\}
\]
is called a partition family.
\end{definition}
We prove the following consequence of Corollary~\ref{cor:dpps_gen}, which asserts that polynomial time counting and sampling is possible for DPPs under partition constraints for constant $p$.

\begin{corollary}\label{cor:partition_sample}
Given a DPP defined by $L\in \R^{m\times m}$ and a partition family $\cC$ with a constant number of parts, there exists a polynomial time sampling algorithm for the distribution
$$\mu_\cC(S) \defeq \frac{\det(L_{S,S})}{\sum_{T\in \cC} \det(L_{T,T})}~~~~~~\mbox{for $S\in \cC$. }$$
\end{corollary}

\begin{proof}
In light of Corollary~\ref{cor:dpps_gen} it suffices to show that every partition family has a succinct representation as a linear family. We show that it is indeed the case.
Consider a partition family $\cC$ induced by the partition $P_1 \cup P_2 \cup \ldots \cup P_p=[m]$ and numbers $b_1, b_2, \ldots, b_p$. Define the following cost vectors: $c_j = 1_{P_j}$, for $j=1,2,\ldots, p$, i.e., the indicator vectors of the sets $P_1, P_2, \ldots, P_p$. Moreover define $K_j$ to be $\{b_j\}$ for every $j=1,2,\ldots, p$. It is then easy to see that ``$c_j(S)\in K_j$'' is implementing the constraint $|P_j \cap S| = b_j$. In other words the family $\cC$ is equal to the linear family defined by cost vectors $c_1, c_2, \ldots, c_p$ and sets $K_1, K_2, \ldots, K_p$.
It remains to observe that $\norm{c_j}_1 = |P_j| \leq m$ and hence $\prod_{j=1}^p \inparen{\norm{c_j}+1} = O(m^p)$. Since $p=O(1)$ the algorithm from Corollary~\ref{cor:dpps_gen} runs in polynomial time.
\end{proof}

\section{Hardness Result}\label{sec:hardnessResults}
In this section we study hardness of $\bcount[\mu, c, C]$. Theorem~\ref{thm:main} implies that $\bcount$ is polynomial time solvable whenever we measure the complexity with respect to the unary encoding length of the cost vector $c$. Here we prove that if $c$ is given in binary, the problem becomes $\mathbf{\#P}-$hard. Moreover, existence of an efficient approximation scheme for a closely related problem (instead of counting all objects of cost {\em  at most} $C$, count objects of cost {\em exactly} $C$) would imply existence of such schemes for counting perfect matchings in non-bipartite graphs (see Section~\ref{sec:hardness_trees}) and for computing mixed discriminants. In both cases, these are notorious open questions and the latter is believed to be unlikely.

\subsection{Mixed Discriminants}\label{sec:mixed_discr}
We relate the  $\bcount$ problem to the well studied problem of computing mixed discriminants of PSD matrices and prove Theorem~\ref{thm:hardness}. Recall the definition:
\begin{definition}\label{def:mixed_disc}
Let $A_1, A_2, \ldots, A_m \in \R^{d\times d}$ be symmetric matrices of dimension $d$. The mixed discriminant of a tuple $(A_1, A_2, \ldots, A_d)$ is defined as
$$D(A_1,A_2, \ldots, A_d)\defeq \frac{\partial^d}{\partial z_1 \ldots \partial z_d}\det(z_1 A_1 +z_2A_2+\cdots+z_dA_d).$$
\end{definition}Computing mixed discriminants of PSD matrices is known to be $\mathbf{\#P}$-hard, since they can encode the permanent. However, as opposed to the permanent, there is no FPRAS known for computing mixed discriminants, and the best polynomial time approximation algorithms by~\cite{Barvinok97,GS02} have an exponentially large approximation ratio.

The main technical component in our proof of Theorem~\ref{thm:hardness} is the following lemma.
\begin{lemma}\label{lemma:mixd}
There is a polynomial time reduction, which given a tuple $(A_1, \ldots, A_n)$ of PSD $n\times n$ matrices outputs a PSD matrix $L\in \R^{m\times m}$, a cost vector $c\in \Z^m$ and a cost value $C\in \Z$ such that
\[
\alpha \cdot D(A_1, A_2, \ldots, A_n) = \sum_{S\subseteq [m],~ c(S)=C} \mu(S),
\]
where $\mu(S) = \det(L_{S,S})$, for $S\subseteq [m]$, and $\alpha$ is an efficiently computable scalar.  Moreover, $\norm{c}_1 \leq 2^{O(n \log n)}$.
\end{lemma}
Before proving Lemma~\ref{lemma:mixd} let us first state several important properties of mixed discriminants, which we will rely on; for proofs of these facts  we refer the reader to~\cite{Bapat89}.

\begin{fact}[Properties of Mixed Discriminants]\label{fact_mixed}
Let $A,B, A_1, A_2, \ldots, A_n$ be symmetric $n\times n$ matrices.
\begin{enumerate}[noitemsep,nolistsep]
\item $D$ is symmetric, i.e., 
\[
D(A_1, A_2, \ldots, A_n)=D(A_{\sigma(1)}, A_{\sigma(2)}, \ldots, A_{\sigma(n)}),~ \text{for any permutation $\sigma \in S_n$}.
\]
\item $D$ is linear with respect to every coordinate, i.e., 
\[
D(\alpha A+\beta B, A_2, \ldots, A_n)=\alpha D(A, A_2, \ldots, A_n)+\beta D(B, A_2, \ldots, A_n).
\]
\item If $A=\sum_{i=1}^n v_i v_i^\top \in \R^{n\times n}$ then we have: $\det(A) = n!~ D(v_1 v_1^\top , \ldots, v_n v_n^\top)$.
\end{enumerate}
\end{fact}

\begin{proofof}{of Lemma~\ref{lemma:mixd}}
Consider a tuple $(A_1, A_2, \ldots, A_n)$ of PSD matrices. The first step is to decompose them into rank-one summands:
$$A_i=\sum_{j=1}^r v_{i,j}v_{i,j}^\top, $$
where $v_{i,j}\in \R^n$ for $1\leq i,j\leq n$ (some $v_{i,j}$'s can be zero if $\mathrm{rank}(A_i)<n$). This step can be performed using the Cholesky decomposition.

Let $M=\{(i,j):1\leq i,j \leq n\}$ and for every $i=1,2, \ldots, n$ define $P_i=\{i\} \times [n]$. We take  $m=|M|=n^2$ and define a family $\cC$  of $n-$subsets of $M$ to be
$$\cC = \{S\subseteq [m]: |S \cap P_i|=1 \mbox{ for every }i=1,2,\ldots,n\}.$$
Let $V$ denote an $m\times n$ matrix with rows indexed by $M$, for which the $e$th row is $v_e$ as above ($e\in M$, i.e.,  $e=(i,j)$ for some $i,j\in [n]$). We also set $L=VV^\top$, hence $L$ is an $m\times m$ symmetric,  PSD matrix.
Finally, let $\mu(S) = \det(L_{S,S}).$ Note that for sets $S$ of cardinality $n$ we have $$\mu(S)=\det(L_{S,S}) =\det(V_S V_S^\top)=\det(V_S^\top V_S) = \det \inparen{\sum_{e\in S} v_e v_e^\top }.$$
In the calculation below we rely on properties of mixed discriminants listed in Fact~\ref{fact_mixed} and on the fact that $|S|=n$ for $S\in \cC$.
\begin{align*}
D(A_1, A_2, \ldots, A_n) &=D\inparen{\sum_{j=1}^n v_{1,j} v_{1,j}^\top, \sum_{j=1}^n v_{2,j} v_{2,j}^\top,\ldots, \sum_{j=1}^n v_{n,j} v_{n,j}^\top}\\
&= \sum_{1\leq j_1, j_2, \ldots, j_n\leq n } D(v_{1,j_1} v_{1,j_1}^\top,v_{2,j_2} v_{2,j_2}^\top,\ldots, v_{n,j_n} v_{n,j_n}^\top )\\
&= \sum_{e_1 \in P_1,e_2 \in P_2, \ldots, e_n\in P_n} D(v_{e_1} v_{e_1}^\top,v_{e_2} v_{e_2}^\top,\ldots, v_{e_n} v_{e_n}^\top )\\
& = \sum_{\{e_1, e_2, \ldots, e_n\} \in \cC} \frac{1}{n!}\det(v_{e_1} v_{e_1}^\top+ v_{e_2} v_{e_2}^\top+\ldots+ v_{e_n} v_{e_n}^\top )
 =\frac{1}{n!}\sum_{S\in \cC}\mu(S).
\end{align*}
It remains to show that the partition family $\cC$  can be represented as
 $\cC=\inbraces{S\subseteq M : c(S) = C}$
 for some cost vector $c\in \Z^M$ and $C\in \Z$, such that $\norm{c}_1=2^{O(n \log n)}$. Indeed, by a reasoning as in Corollary~\ref{cor:partition_sample} we can represent $\cC$ as a linear family with $n$ constraints of the form $c^{(i)}(S)=1$ for $i=1,2,\ldots, n$ and $c^{(i)}\in \{0,1\}^{n\times n}$. It is not hard to see that these can be combined into one constraint $c(S)=C$ with $\norm{c}_1=(n^2)^{n+O(1)}=2^{O(n \log n)}.$
Now, it remains to observe that the scalar $\alpha$ from the statement of the lemma is  $n!$ and the steps of the reduction are  efficient.
\end{proofof}
\begin{proofof}{of Theorem~\ref{thm:hardness}}
In light of Lemma~\ref{lemma:mixd}, the problem of computing $\sum_{S\subseteq [m], c(S)=C} \mu(S)$ for determinantal functions $\mu$ is at least as hard as computing mixed discriminants. The $\bcount$ problem is very similar, with the only difference that it is computing the sum over all sets of cost $c(S)$ at most $C$. However, clearly by solving the $\bcount$ problem for $C$ and $C-1$ one can compute $\sum_{S\subseteq [m], c(S)=C} \mu(S)$ by just subtracting the obtained results.
\end{proofof}

\section{Mixed Discriminants and Mixed Characteristic Polynomials}\label{sec:mixed}
{\it Mixed Characteristic Polynomials}  played a crucial role in the proof of the Kadison-Singer conjecture. Making this proof algorithmic is an outstanding open question that naturally leads to the problem of computing the maximum root of these mixed characteristic polynomials. In this section, we show how Corollary~\ref{cor:partition_sample} implies a polynomial time algorithm for higher-order coefficients of such polynomials.
We start by defining mixed characteristic polynomials. We use the following simplified notation for partial derivatives: $\partial_{x_i} f(x)$ is an abbreviation for $\frac{\partial}{\partial x_i} f(x)$.
\begin{definition}\label{def:mixed}
Let $A_1, A_2, \ldots, A_m \in \R^{d\times d}$ be symmetric matrices of dimension $d$.
The mixed characteristic polynomial of $A_1, A_2, \ldots, A_m$ is defined as
$$\mu[A_1,\ldots,A_m](x) \defeq \prod\limits_{i=1}^{m} \inparen{1 - \partial_{z_i}}\det\bigg(xI + \sum\limits_{i=1}^{m}z_i A_i\bigg)\bigg\rvert_{z_1 = \cdots = z_m = 0}.$$
\end{definition}
Note in particular that while mixed discriminants are defined for a tuple whose length matches the dimension $d$ of the matrices,  for the case of mixed characteristic polynomials the number $m$ can be arbitrary. In fact, when $m=d$, the constant term in the mixed characteristic polynomial is (up to sign) equal to the mixed discriminant of the input tuple.

However, one may wonder whether all of the coefficients in these polynomials are hard to compute. The following result shows that higher-degree coefficients are computable in polynomial time. Roughly, the proof relies on the observation that the higher-degree coefficients in the mixed characteristic polynomial are sums of mixed discriminants that only have constantly many \textit{distinct} matrices. As we demonstrate, computing such mixed discriminants reduces to counting for DPPs under partition constraints with a constant number of parts, which allows us to apply Corollary~\ref{cor:partition_sample}.
The formal statement of the theorem follows \footnote{Independent of our work which first appeared in \cite{DKSV}, a recent preprint \cite{AGSS17} devise a different algorithm to obtain a similar result}.

\begin{theorem}\label{thm:higherOrderThm} Given a set of $m$ symmetric, PSD matrices $A_1, \ldots, A_m \in \mathbb{R}^{d \times d}$, one can compute the coefficient of $x^{d-k}$ in $\mu[A_1,\ldots,A_m](x)$, in $\poly(m^k) $ time.
\end{theorem}

\noindent
An important component in the proof of Theorem~\ref{thm:higherOrderThm} is a reduction from counting for partition constrained DPPs to mixed discriminants. In fact we use it as a subroutine for computing higher-order coefficients of the mixed characteristic polynomial. In Section \ref{sec:hardnessResults} we provided a reduction in the opposite direction, thus establishing an \textit{equivalence} between mixed discriminants and counting for partition constrained DPPs.

\begin{lemma}\label{lemma:PCtoMD} Given a set of $m$ vectors $v_1,\ldots,v_m \in \R^r$ and a partition of $[m]=P_1\cup\cdots\cup P_p$ into disjoint, non-empty sets, consider a partition family $\cC= \{S\subseteq [m]: |S \cap P_j|=b_j \mbox{ for every }j=1,2,\ldots, p\}$ such that $\sum_{j=1}^{p} b_j = r$.  Let $(A_1,\ldots, A_r)$ be an $r$-tuple  of PSD $r \times r$ matrices such that $\inparen{A_1, A_2, \ldots, A_r} = (\overbrace{B_1, \ldots, B_1}^{b_1 \mbox{ times }}, \overbrace{B_2, \ldots, B_2}^{b_2 \mbox{ times }}, \ldots, \overbrace{B_p, \ldots, B_p}^{b_p \mbox{ times }})$ where $B_i = \sum\limits_{e \in P_i} v_e v_e^\top$ for every partition $P_i$, the following equality holds:
$$\prod_{i=1}^{p} b_i! \cdot D(A_1, A_2, \ldots, A_r) = \sum_{S \in \cC} \det(V_{S}V_{S}^\top),$$
where $V\in \R^{m \times r}$ denotes the matrix formed by arranging the vectors $v_1, \ldots, v_m$ row-wise.
\end{lemma}
\begin{proof} Consider the quantities $B_i$ and $(A_1,A_2,\ldots, A_r)$ as defined in the theorem.
By applying linearity multiple times to all coordinates of $D(A_1, A_2, \ldots, A_r)$ we find that:
$$D(A_1, A_2, \ldots, A_r) = \alpha \sum_{S\in \cB}D(v_{e_1}v_{e_1}^\top, v_{e_2}v_{e_2}^\top, \ldots, v_{e_r}v_{e_r}^\top),$$
where $S$ is $\{e_1, e_2, \ldots, e_r\}$ in the summation above and $\alpha$ is $\prod_{i=1}^p b_i!$.
 This is because \\$D(v_{e_1}v_{e_1}^\top, v_{e_2}v_{e_2}^\top, \ldots, v_{e_r}v_{e_r}^\top)=0$ whenever $e_1, e_2, \ldots, e_r$ are not pairwise distinct.
 We use Fact~\ref{fact_mixed} again to obtain  that
$$D(v_{e_1}v_{e_1}^\top, v_{e_2}v_{e_2}^\top, \ldots, v_{e_r}v_{e_r}^\top) = \frac{1}{r!} \det(v_{e_1}v_{e_1}^\top+ v_{e_2}v_{e_2}^\top+ \ldots+ v_{e_r}v_{e_r}^\top) = \det(V_S V_S^\top).$$
This concludes the proof.
Furthermore, it is evident that the $r$-tuple $(A_1, A_2, \ldots, A_r)$ is efficiently computable given the partition family $\cC$ and matrix $V$.

\end{proof}
\begin{proofof}{of Theorem \ref{thm:higherOrderThm}}
First note that without loss of generality we can assume that $d\leq m$, as otherwise -- if $d>m$ we can add $(d-m)$ zero-matrices which does not change the result but places us in the $d\leq m$ case.
 The starting point of our proof is an observation made in \cite{MSS13} which provides us with another expression for the mixed characteristic polynomial in terms of mixed discriminants:
\begin{equation}\label{mixedcarpolyeqn}
\mu[A_1,\ldots,A_m](x) = \sum\limits_{k=0}^{d} x^{d-k} (-1)^{k} \sum\limits_{S \in {[m] \choose k}} D( (A_i)_{i \in S} )
\end{equation}
where we denote $D(A_1,\ldots,A_k) = \frac{1}{(d-k)!}D(A_1,\ldots,A_k,I, \ldots,I)$ with the identity matrix $I$ repeated $d-k$ times.
Therefore, our task reduces to computing $O(m^k)$ mixed discriminants of the form $D(A_1, \ldots, A_k,I, \ldots, I)$. Below we show that such a quantity is computable in $\poly(d^k)$ time which  concludes the proof.

Consider the Cholesky decomposition of $A_i$ for $i=1,2,\ldots, k+1$ (we set $A_{k+1}=I$ for convenience)
$$ A_i = \sum_{j=1}^{d} u_{i,j} u_{i,j}^\top.$$
Let $M=\{(i,j) : 1 \leq i \leq k+1, 1 \leq j \leq d\}$ be the ground set of a partition family of size $m \defeq (k+1)d.$
Define an $m\times d$ matrix $U$ by placing $u_{i,j}$'s as rows of $U$.

Further, consider a partition $M=P_1\cup\cdots\cup P_{k+1}$ with $P_i = \{i\} \times [d]$ for all $i=1,\ldots,k+1$
and let $b_1=\ldots=b_k=1$ and $b_{k+1}=d-k$. This gives rise to a partition family
$$\cC = \{T \in M : |T \cap P_i|=b_i \mbox{ for all } i=1,\ldots,k+1\}.$$
We claim that
\begin{equation}\label{eq:mix_partition}
\prod\limits_{i=1}^{k+1} b_i!\sum\limits_{T \in \cC} \det(U_T U_T^\top) = D( A_1, \ldots, A_k,I\ldots,I).
\end{equation}
This follows from Lemma \ref{lemma:PCtoMD} by considering this partition family $\cC$ and matrix $U$ as defined here.
Equation~\eqref{eq:mix_partition} combined with the counting result for DPPs under partition constraints (Corollary~\ref{cor:partition_sample}) conclude the proof.

\end{proofof}
The second observation is more general in its nature and tries to answer the question whether computing mixed characteristic polynomials is strictly harder than computing mixed discriminants. In fact, as noted above, the coefficients of mixed characteristic polynomials are expressed as sums of (an exponential number of) mixed discriminants. We show that these exponential sums can be computed by evaluating a \textit{single} mixed discriminant of matrices of size at most $d+n$. Moreover, our reduction is \textit{approximation-preserving}, hence demonstrating that approximating mixed discriminants are computationally equally hard as approximating the coefficients of the mixed characteristic polynomials.  We remark that our reduction can be thought of as a generalization of a result for approximating the number of $k$-matchings in a bipartite graph (\cite{FriedlandLevy06}).
\begin{theorem}\label{thm:MCPreduction} Given a tuple of $m$ symmetric, positive semi-definite matrices $A_1,\ldots, A_m \in \R^{d \times d}$ with $d\leq m$ and $k \in \{1,\ldots,d\}$, there exist a tuple of $m+d-k$ symmetric, positive semi-definite matrices $B_1,\ldots,B_{m+d-k} \in \R^{(m+d-k) \times (m+d-k)}$ such that the coefficient of $x^{d-k}$ in the mixed characteristic polynomial $\mu[A_1,\ldots,A_m](x)$,
\begin{align*}
\sum\limits_{S \in {[m] \choose k}} D( (A_i)_{i\in S} ) = \frac{1}{(m-k)!(d-k)!}D(B_1,\ldots,B_{m+d-k})
\end{align*}
\end{theorem}
\begin{proof} We first show how to construct the $m+d-k$ matrices $B_1,\ldots,B_{m+d-k}$ from $A_1,\ldots,A_m$. The matrices $B_1,\ldots,B_{m+d-k}$ that we consider are $2$-by-$2$ block diagonal matrices that we construct by taking appropriate direct sums. Recall that the direct sum of two matrices $A$ and $B$ of size $d_1 \times d_1$ and $d_2 \times d_2$ is a matrix of size $(d_1 + d_2) \times (d_1 + d_2)$ defined as
\setlength{\arrayrulewidth}{.5pt}
$$
G=\left[
\begin{array}{c|c}
A & \mathbf{0}_{d_1 \times d_2}\\ \hhline{-|-}
 \mathbf{0}_{d_2 \times d_1}& B \\ \hhline{~|~}
\end{array}\right]
$$
where $\mathbf{0}_{m \times n}$ is an $m$-times-$n$ matrix consisting of all zeros.
We define the first $m$ matrices to be direct sums of the $A_i$ matrices with the identity matrix of order $m-k$, i.e., $I_{m-k}$ and the remaining $d-k$ matrices to all be equal to the direct sum of the identity matrix of order $d$, i.e., $I_d$ with the square zero matrix of order $m-k$, i.e., $\mathbf{0}_{m-k}$. Formally,
\begin{align*}
B_i = \begin{cases}
A_i \oplus I_{m-k} & \mbox{for } i \in \{1,\ldots,k\},\\
I_d \oplus \mathbf{0}_{m-k} & \mbox{otherwise}
\end{cases}
\end{align*}
We now proceed to prove the claim of the theorem from the definition of the mixed discriminant in Definition \ref{def:mixed}. For any subset $S \subseteq [m]$, denote $\partial^S = \prod_{i \in S} \partial_{z_i}$.
\begin{align*}
D(B_1,\ldots,B_{m+d-k})&=\partial_{z_1}\ldots\partial_{z_{m+d-k}}\det(z_1 B_1 + \ldots + z_{m+d-k} B_{m+d-k})\\
&= \partial_{z_1}\ldots\partial_{z_{m+d-k}}\det\bigg(\left[
\begin{array}{c|c}
\sum\limits_{i=1}^{m} z_i A_i + \sum\limits_{i=1}^{d-k}z_{m+i} I_d & \mathbf{0}_{d \times (m-k)}\\ \hhline{-|-}
 \mathbf{0}_{(m-k) \times d }& \sum\limits_{i=1}^{m} z_i I_{m-k} \\ \hhline{~|~}
\end{array}\right]\bigg)\\
& = \partial_{z_1}\ldots\partial_{z_{m+d-k}} (z_1 + \ldots + z_m)^{m-k} \det(\sum_{i=1}^{m} z_i A_i + \sum_{i=1}^{d-k} z_{m+i} I_d)\\
&= \sum_{\substack{S \subseteq [m] \\ |S| = m-k} } [\partial^S (z_1 + \ldots + z_m)^{m-k}] [\partial^{S^c} \partial_{z_{m+1}}\ldots \partial_{z_{m+d-k}} \det(\sum_{i=1}^{m} z_i A_i + \sum_{i=1}^{d-k} z_{m+i} I_d)] \\
\end{align*}
\begin{align*}
& = \sum_{\substack{S \subseteq [m] \\ |S| = m-k} } (m-k)! \partial^{S^c} \partial_{z_{m+1}}\ldots \partial_{z_{m+d-k}}\det(\sum_{i \in S^c} z_i A_i + (z_{m+1} + \ldots z_{m+d-k}) I_d)\\
& = (m-k)! \sum_{\substack{S \subseteq [m] \\ |S| = k}} D( (A_i)_{i \in S}, \overbrace{I,\ldots,I}^{d-k \mbox{ times}} )\\
& = (m-k)! (d-k)! \sum_{\substack{S \subseteq [m] \\ |S| = k}} D( (A_i)_{i \in S} )
\end{align*}
The fourth to last equality follows simply from chain rule. Since we have an equality in the expression, the reduction is clearly approximation preserving and we are done.
\end{proof}
\noindent
The above theorem in particular allows us to compute in polynomial time, the mixed characteristic polynomial \textit{exactly}, when the linear matrix subspace spanned by the input matrices has constant dimension. This follows by combining Theorem \ref{thm:MCPreduction} with Theorem 5.1 in~\cite{Gurvits05}.

\begin{corollary}\label{cor:rank} Suppose $A_1, A_2, \ldots, A_m \in \R^{d\times d}$ span a linear space of dimension $k$, then there exists a deterministic algorithm to compute $\mu[A_1, \ldots, A_m](x)$ in $\poly(m^k)$ time.
\end{corollary}
\begin{proof}
In the proof of Theorem \ref{thm:MCPreduction}, the mixed discriminants computed are not of $A_1, \ldots, A_m$ but rather are of modified matrices. However, it is easy to see that for all tuples on which mixed discriminant is called, the dimension of the linear space spanned by them is at most $k+1$. It is proved in~\cite{Gurvits05} that such mixed discriminants can be computed in $O(m^{2k+2})$ time.
\end{proof}

\section{Budget-Constrained Sampling and Counting for Regular Matroids}\label{sec:matroids}
Consider the following problem: given an undirected graph $G$ with weights $c\in \R^m$ on its edges, sample a uniformly random spanning tree of cost at most $C$ in $G$. This generalizes the problem of sampling uniformly random spanning trees~ \cite{pemantle2004uniform} and sampling a random spanning tree of minimum cost~ \cite{Eppstein95}. Below we study the generalized version of this problem by considering regular matroids, indeed spanning trees arise as bases of the graphic matroid, which is known to be regular. We prove that the counting and sampling problem in this setting can be solved efficiently whenever  $c$ is polynomially bounded.

\begin{theorem}[Counting and Sampling Bases of Matroids]\label{thm:reg_matroid}
Let $\cM$ be a regular matroid on a ground set $[m]$ with a set of bases $\mathcal{B}$. There exists a counting algorithm which, given a cost vector $c\in \Z^m$ and a value $C\in \Z$, outputs the cardinality of the set $\{S\in \cB: c(S) \leq C\}$ and a sampling algorithm which,  given a cost vector $c\in \Z^m$ and a value $C\in \Z$, outputs a random element in the set $\{S\in \cB: c(S) \leq C\}$.
The running time of both algorithms is polynomial in $m$ and $\norm{c}_1$.
\end{theorem}

\begin{proofof}{of Theorem~\ref{thm:reg_matroid}}
Let $\cM \subseteq 2^{[m]}$ be a regular matroid and $\cB \subseteq 2^{[m]}$ be its set of bases. We prove that the generating polynomial $\sum_{S\in \cB} x^S$ is efficiently computable.  We use the characterization of regular matroids as those which can be linearly represented by a totally unimodular matrix. In other words, there exists a totally unimodular matrix $A \in \Z^{m \times d}$ such that if we denote by $A_e\in \Z^d$ the $e^{th}$ row of $A$ it holds that:
\begin{equation}\label{eq:repr}
S\in \cM ~~\Leftrightarrow ~~ \{A_e: e\in S\} \mbox{ is linearly independent. }
\end{equation}
Let $r\leq d$ be the rank of the matroid $\cM$, i.e., the cardinality of any set in $\cB$. We claim that without loss of generality one can assume that $d=r$. Indeed, we prove that there is a submatrix $A' \in \Z^{m \times r}$ of $A$, such that~\eqref{eq:repr} still holds with $A$ replaced by $A'$. To this end suppose that $d>r$. It is easy to see that the rank of $A$ is $r$, otherwise, by~\eqref{eq:repr} there would be a set $S$ of cardinality at least $r+1$ with $S\in \cM$. Hence there is a column in $A$ which is a linear combination of the remaining columns, we can freely remove this column from $A$, while~\eqref{eq:repr} will be still true. By doing so, we finally obtain a matrix $A'$ with exactly $r$ rows, which satisfies~\eqref{eq:repr}.

By the fact that $A$ has $r$ columns we have:
\begin{equation}\label{eq:bases}
S\in \cB ~~\Leftrightarrow ~~ A_S \mbox{ is nonsingular,}
\end{equation}
where by $A_S$ we mean the $|S| \times r$ submatrix of $A$ corresponding to rows from $S$. In particular, for a set $S\subseteq [m]$ of cardinality $r$ we have:
\begin{equation}\label{eq:unim}
S\in \cB ~~\Leftrightarrow ~~ \det(A_S)\neq 0 ~~ \Leftrightarrow ~~ \det(A_S^\top A_S)=1,
\end{equation}
where the last equivalence follows from $A$ being totally unimodular. Let us now consider the polynomial
$$g(x_1, x_2, \ldots, x_m) = \det\inparen{\sum_{e=1}^m x_e A_e A_e^\top}.$$
By the Cauchy-Binet theorem we obtain:
$$g(x_1, x_2, \ldots, x_m) = \sum_{|S|=r} \det\inparen{\sum_{e\in S} x_e A_e A_e^\top}= x^S \det(A_S^\top A_S).$$
In other words, $g$ is equal to $g_\mu$ -- the generating polynomial of the function $\mu: 2^{[m]} \to \R$ given by
$$\mu(S) = \begin{cases}
1 & \mbox{if $S\in \cB$}\\
0 & \mbox{otherwise.}
\end{cases}$$
Therefore, since $g_\mu$ is efficiently computable, by Theorem~\ref{thm:main} the $\bcount[\mu, c, C]$ is efficiently solvable. This fact, together with Theorem~\ref{thm:equiv_count_sample} imply that sampling also can be made efficient.
\end{proofof}

\section{Hardness for Spanning Trees}\label{sec:hardness_trees}
We show that $\bcount$ is at least as hard as counting perfect matchings in a non-bipartite graph. The proof  relies on a combinatorial reduction from counting perfect matchings in a graph to counting budget constrained spanning trees.

\begin{theorem}\label{thm:trees}
There is a polynomial time reduction which given a graph $G=(V,E)$ with $n$ vertices and $m$ edges outputs a graph $G'$ with $n$ vertices and $O(m+n^2)$ edges, a cost vector $c\in \N^m$ with $\norm{c}_1\leq  2^{O(m \log m)}$ and a value $C\in \N$, such that:
$$PM(G) = \alpha \cdot ST_C(G')$$
where $PM(G)$ denotes the number of perfect matchings in $G$, $ST_C(G')$ denotes the number of spanning trees of total cost $C$ in $G'$ and  $\alpha=\frac{n^2}{2}(2n)^{-n/2}$.
\end{theorem}

\begin{proof}
Let $G=(V,E)$ be an undirected graph, let $n=|V|$ and $m=|E|$. We construct a new graph $G'$ and a cost vector $c$, such that counting perfect matchings in $G$ is equivalent to counting spanning trees of specified cost $C$ in $G'$ .

 The graph $G'=(V,E')$ is obtained by adding a complete graph to $G$, i.e.,   ${n \choose 2}$ edges, one between every pair of vertices. We call the set of new edges $F$, hence $E'=E\cup F$. Note that $E'$ is a multiset. To all edges $e\in F$ we assign cost $c_e=0$, while for the original edges the costs are positive and defined below.

Let $b=m'+1$, where $m'=|E'|$ is the number of edges in $G'$. We define the cost of an edge $e=ij \in E$ to be:
$$c_e = b^i +b^j.$$
Note that from the choice of $b$ and $c$ it follows that given a cost $c(S)$ of some set $S\subseteq E$, we can exactly compute how many times a given vertex appears as an endpoint of an edge in $S$. Indeed, if we have:
$$c(S) = \sum_{i=1}^n \delta_i b^i$$
such that $0\leq \delta_i \leq b-1$ (the $b-$ary  representation of $c(S)$),
then the degree of vertex $i$ in $S$ is $\delta_i$. This follows from the fact that $b$ is chosen to avoid carry overs when computing $c(S)$ in the $b-$ary numerical system.
Therefore, it is now a natural choice to define $C \defeq \sum_{i=1}^n b^i.$ We claim that every perfect matching in $G$ corresponds to exactly $\alpha=\frac{n^2}{2}(2n)^{-n/2}$ different spanning trees of cost $C$ in $G'$.

To prove this claim, fix any spanning tree $S$ of cost $c(S)=C$. Note first that we have $c(S\cap E)=c(S)$ because all of the edges $e\notin E$ have cost $0$. Moreover, the set $M \defeq S\cap E$ is a perfect matching in $G$, because $c(M)=C$ implies that the degree of every vertex in $M$ is one. It remains to show that every perfect matching $M$ in $G$ corresponds to exactly $\alpha$ spanning trees of cost $C$ in $G$.

Fix any perfect matching $M_0$ in $G$. We need to calculate how many ways are there to add $\frac{n}{2} - 1$ edges from $E'$ to obtain a spanning tree of $G'$. By contracting the matching $M_0$ to $\frac{n}{2}$ vertices and considering edges in $E'$ only, we obtain a complete graph on $\frac{n}{2}$ vertices with $4$ parallel edges going between every pair of vertices. The answer is the number of spanning trees of the obtained graph. Cayley's formula easily implies that this number is  $4^{\frac{n}{2}-1}\inparen{\frac{n}{2}}^{\frac{n}{2}-2}$ which equals $\alpha^{-1}$.
  \end{proof}

\section{Equivalence Between Counting and Sampling }\label{sec:equiv}

In this section we state and prove a theorem that implies that the $\ccount[\mu, \cC]$ and $\csample[\mu, \cC]$ problems are essentially equivalent.
We prove that, for a given type of constraints $\cC$, a polynomial time algorithm for counting can be transformed into a polynomial time algorithm for sampling and vice versa.
This section follows the convention that $\mu: 2^{[m]} \to \R_{\geq 0}$ is any function that assigns nonnegative values to subsets of $[m]$ and $\cC \subseteq 2^{[m]}$ is any family of subsets of $[m]$.

\begin{theorem}[Equivalence Between Approximate Counting and Approximate Sampling]\label{thm:equiv_count_sample}
Consider any function $\mu: 2^{[m]} \to \R_{\geq 0}$ and a family $\cC$ of subsets of $[m]$.
Let $\mu_\cC:\cC \to [0,1]$ be a distribution over $S\in \cC$ such that $\mu_\cC(S) \propto \mu(S)$. We assume evaluation oracle access     to the generating polynomial $g_\mu$ of $\mu$, and define the following two problems:
\begin{itemize}
\item   {\bf Approximate $\cC$-sampling:}  given a precision parameter $\eps>0$, provide a sample $S$ from a distribution $\rho : \cC \to [0,1]$ such that $\norm{\mu_\cC - \rho}_1<\eps$.
\item   {\bf Approximate $\cC$-counting:}  given a precision parameter $\eps>0$, output a number $X\in \R$ such that $X(1+\eps)^{-1}\leq \sum_{S\in \cC} \mu(S) \leq X(1+\eps).$
\end{itemize}
The time complexities of the above problems differ by at most a multiplicative factor of $\poly(m, \eps^{-1})$.
\end{theorem}
\begin{remark}
Note that the above theorem establishes equivalence between \emph{approximate} variants of $\ccount[\mu, \cC]$ and $\csample[\mu, \cC]$. This is convenient for applications, because the exact counting variants of these problems are often $\mathbf{\#P}-$hard. Still, for some of them, efficient approximation schemes are likely to exist.
Further, we mention that the implication from exact counting to exact sampling holds, hence the sampling algorithms that we obtain in this paper are exact.
\end{remark}
Theorem~\ref{thm:equiv_count_sample} follows from a self-reducibility property \cite{JVV86} of the counting problem.   Before we present the proof of Theorem~\ref{thm:equiv_count_sample}, we introduce some terminology and state assumptions for the remaining part of this section.
The function $\mu:2^{[m]}\to \R_{\geq 0}$ is given as an evaluation oracle for $g_\mu(x)=\sum_{S\subseteq[m]} \mu(S) x^S$.
In particular, we measure complexity with respect to the number of calls to such an oracle.
An algorithm which, for a fixed family $\cC \subseteq 2^{[m]}$ and every function $\mu$, given access to $g_\mu$ computes $\sum_{S\in \cC} \mu(S)$ is called a $\cC$-counting oracle. Similarly, we define a $\cC$-sampling oracle to be an algorithm which, given access to $g_\mu$, provides samples from the distribution
$$\mu_\cC(S) \defeq \frac{\mu(S)}{\sum_{T\in \cC} \mu(T)} ~~~~~~\mbox{ for $S\in \cC$.}$$

\subsection{Counting Implies Sampling}

We now show how counting implies sampling; the reverse direction is presented in Appendix~\ref{sec:samplingtocounting}.  It proceeds by inductively conditioning on certain elements not being in the sample. For this idea to work one has to implement conditioning using the $\cC-$sampling oracle and access to the generating polynomial only.
Below we state the implication from counting to sampling in the exact variant. The approximate variant also holds, with an analogous proof.
\begin{lemma}[Counting Implies Sampling]\label{lem:count_sample}
Let $\cC$ denote a family of subsets of $[m]$. Suppose access to a $\cC$-counting oracle is given. Then, there exists a $\cC$-sampling oracle which, for any function $\mu: 2^{[m]} \to \R_{\geq 0}$, makes $\poly(m)$ calls to the counting oracle and to $g_\mu$ and outputs a sample from the distribution $\mu_\cC$.
\end{lemma}

\begin{proof}
Let $\mathbf{S}$ be the random variable corresponding to the sample our algorithm outputs; our goal is to have $\mathbf{S} \sim \mu_\cC$.
The sampling algorithm proceeds as follows: It sequentially considers each element $e\in [m]$ and tries to decide (at random) whether to include $e \in \mathbf{S}$ or not.
To do so, it first computes the probability $\Pb(e\in \mathbf{S})$ \emph{conditioned on all decisions thus far}. It then flips a biased coin with this probability, and includes $e$ in $\mathbf S$ according to its outcome. More formally, the sampling algorithm can be described as follows:

\begin{enumerate}[noitemsep,nolistsep]
\item Input: $V\in \R^{m\times r}$, a number $k\leq r$.
\item Initialize: $Y=\emptyset$, $N=\emptyset$.
\item For $e=1,2,\ldots,m:$
\begin{enumerate}[noitemsep,nolistsep]
\item Compute the probability $p=\Pb(e\in \mathbf{S} : Y\subseteq \mathbf{S}, N\cap \mathbf{S}=\emptyset)$ under the distribution $\mathbf{S} \sim \mu_\cC$.
\item Toss a biased coin with success probability $p$. In case of success add $e$ to the set $Y$, otherwise add $e$ to $N$.
\end{enumerate}
\item Output: $\mathbf{S}=Y.$
\end{enumerate}
It is clear that the above algorithm correctly samples from $\mu_\cC$. It remains to show that $\Pb(e\in S : Y\subseteq S, N\cap S=\emptyset)$ can be computed efficiently. This follows from Lemma~\ref{lemma:conditioning} below.
  \end{proof}

\begin{lemma}\label{lemma:conditioning}
Let $Y$ and $N$ be disjoint subsets of $[m]$ and consider any $e\in [m]$.
Suppose $\mathbf{S}$ is distributed according to $\mu_\cC$. If we are given access to a $\cC$-counting oracle and to $g_\mu$, then $\Pb(e\in \mathbf{S} : Y\subseteq \mathbf{S}, N\cap \mathbf{S}=\emptyset)$ can be computed in $\poly(m)$ time.
\end{lemma}
\begin{proof}
Assume $e\in [m]\setminus (Y\cup N)$; otherwise the probability is clearly 0 or 1.
Let $Y^\prime =Y\cup \{e\}$, then
$$\Pb(e\in \mathbf{S} : Y\subseteq \mathbf{S}, N\cap \mathbf{S}=\emptyset) = \frac{\sum_{S\in \cC, Y^\prime \subseteq S, N\cap S=\emptyset}\mu(S)}{\sum_{S\in \cC,Y\subseteq S, N\cap S=\emptyset}\mu(S)}.$$
We now show how to compute such sums:
Introduce a new variable $y$, and for every $e\in [m]$ define:
$$w_e \defeq \begin{cases}
yx_e & \mbox{for }e\in Y,\\
0 \quad & \mbox{for } e\in N,\\
x_e & \mbox{otherwise.}
\end{cases}$$
We interpret the expression
$g_\mu(w_1, w_2, \ldots, w_m)$
as a generating polynomial for a certain function $\mu'(y) : 2^{[m]} \to \R$; i.e.,
$$g_{\mu'}(x) \defeq g_\mu(w_1, w_2, \ldots, w_m) =  \sum_{ S\cap N=\emptyset} y^{|S\cap Y|} x^S \mu(S).$$
Define a polynomial
$$h(y) \defeq \sum_{S\in \cC, S\cap N=\emptyset} y^{|S\cap Y|} \mu(S).$$
It follows that $h(y)$ is a polynomial of degree at most $|Y|$. In fact, the sum we are interested in is simply the coefficient of $y^{|Y|}$ in $h(y)$. The last thing to note is that we can compute $h(y)$ exactly by evaluating it for $|Y|+1$ different values of $y$ and then performing interpolation. Hence, we just need to query the $\cC$-counting oracle $\inparen{|Y|+1}$ times giving it $\mu'$ as input (for various choices of $y$).\footnote{The provided argument does not generalize directly to the case when the counting oracle is only approximate (because of the interpolation step). However, as we need to compute the top coefficient of a polynomial $h(y)$ only, we can alternatively do it by evaluating $h(y)$ and dividing by $y^d$ (for $d=\deg(h)$) at a very large input $y\in \R$.}
\end{proof}

\subsection{Sampling Implies Counting}
\label{sec:samplingtocounting}

We show the implication from sampling to counting in Theorem~\ref{thm:equiv_count_sample}. Similarly as for the opposite direction we assume for simplicity that the sampling algorithm is exact, i.e.,  we prove the following lemma. The approximate variant holds with an analogous proof.
\begin{lemma}[Sampling Implies Counting]\label{lem:sample_count}
Let $\cC$ denote a family of subsets of $[m]$. Suppose we have access to a $\cC$-sampling oracle. Then, there exists a $\cC$-counting oracle which for any input function $\mu: 2^{[m]} \to \R$ (given as an evaluation oracle for $g_\mu$) and for any precision parameter $\eps>0$ makes $\poly(m,\nfrac{1}{\eps})$ calls to the sampling oracle, and approximates the sum:
$$\sum_{S\in \cC} \mu(S)$$
within a multiplicative factor  of $(1+\eps)$. The algorithm has failure probability exponentially small in $m$.
\end{lemma}
Let us first state the algorithm which we use to solve the counting problem. Later in a sequence of lemmas we explain how to implement it in polynomial time and reason about its correctness. In the description, $\mathbf{S}$ denotes a random variable distributed according to $\mu_\cC$.

\begin{enumerate}[noitemsep,nolistsep]
\item Initialize $U \defeq [m]$, $X \defeq 1$.
\item Repeat
\begin{enumerate}[noitemsep,nolistsep]
\item Estimate the probability  $\Pb(\mathbf{S}=U : \mathbf{S} \subseteq U)$, if it is larger than $(1-\frac{1}{m})$, terminate the loop.
\item  Find an element $e\in U$ so that $\Pb(e\notin \mathbf{S} : \mathbf{S}\subseteq U)\geq \frac{1}{m^2}$.
\item Approximate $p_e \defeq \Pb(e\notin \mathbf{S} : \mathbf{S}\subseteq U)$ up to a multiplicative factor $\frac{\eps}{m}$.
\item Update $X \defeq X\cdot \rho_e$, where $\rho_e$ is the estimate for $p_e$.
\item Remove $e$ from $U$, i.e., set $U \defeq U\setminus \{e\}$.
\end{enumerate}
\item Return $X\cdot \mu(U)$.
\end{enumerate}

\begin{lemma}\label{lemma_additive}
Given $U\subseteq [m]$ and $e\in U$, assuming access to a $\cC$-sampling oracle, we can approximate the quantity
$$p_e=\Pb(e\notin \mathbf{S} : \mathbf{S}\subseteq U)$$
where $\mathbf{S}$ is distributed according to $\mu_\cC$, up to an additive error $\delta>0$ in time $\frac{\poly(m)}{\delta^2}.$ The probability of failure can be made $\frac{1}{m^c}$ for any $c>0$.
\end{lemma}
\begin{proof} We sample a set $S\in \cC$ from the distribution $\Pb(S)\propto \mu(S)$ conditioned on $S\subseteq U$. This can be done using the sampling oracle, however instead of sampling with respect to $\mu$ one has to sample with respect to a modified function $\mu'$ which is defined as $\mu'(S)=\mu(S)$ for $S\subseteq U$ and $\mu'(S)=0$ otherwise. Note that the generating polynomial for $\mu'$ can be easily obtained from $g_\mu$ by just plugging in zeros at positions outside of $U$. Given a sample $S$ from $\mu'$ we define
$$X=\begin{cases}
1\qquad \mbox{ if } e\notin S,\\
0 \qquad \mbox{ otherwise. }
\end{cases}$$
Repeat the above independently $N$ times, to obtain $X_1, X_2, \ldots, X_N$ and finally compute the estimator:
$$Z=\frac{X_1+X_2+\cdots+ X_N}{N}.$$
By Chebyshev's inequality, we have:
$$\Pb(|Z-p_e|\geq \delta) \leq \frac{1}{N\delta^2}$$
Thus, by taking $N=\frac{\poly(m)}{\delta^2}$ samples, with probability $\geq 1 - \frac{1}{\poly(m)}$ we can obtain an additive error of at most $\delta$.
  \end{proof}

\begin{lemma}\label{lemma:largeU}
If $U\subseteq [m]$ is such that $\Pb(\mathbf{S}=U : \mathbf{S}\subseteq U)\leq (1-\frac{1}{m})$ then there exists an element $e\in U$ such that $\Pb(e\notin \mathbf{S} : \mathbf{S} \subseteq U)\geq \frac{1}{m^2}$, where $\mathbf{S}$ is distributed according to $\mu_\cC$.
\end{lemma}

\begin{proof}
Let  $\mathbf{T}$ be the random variable $\mathbf{S}$ conditioned on $\mathbf{S} \subseteq U$. Denote $q_e=\Pb(e\in \mathbf{S} : \mathbf{S} \subseteq U)$, we obtain
$$\sum_{e\in U} q_e=\mathbb{E}(|\mathbf{T}|)\leq \inparen{1-\frac{1}{m}} |U| + \frac{1}{m}\inparen{|U|-1}=|U|-\frac{1}{m}.$$
The inequality in the above expression follows from the fact that the worst case upper bound would be achieved when the probability of $|\mathbf{T}| = |U|$ is \textit{exactly} $1-\frac{1}{m}$ and with the remaining probability, $|\mathbf{T}| = |U|-1$. Hence $\sum_{e\in U} (1-q_e)\geq \frac{1}{m}$, which implies that $(1-q_e)\geq \frac{1}{m^2}$ for some $e\in U$.
  \end{proof}

We are now ready to prove Lemma~\ref{lem:sample_count}.

\begin{proof} (of Lemma \ref{lem:sample_count})
We have to show that the algorithm given above can be implemented in polynomial time and it gives a correct answer.

Step 2(a) can be easily implemented by taking $\poly(m)$ samples conditioned on $\mathbf{S} \subseteq U$ (as in the proof of Lemma~\ref{lemma_additive}). This gives us an approximation of $q_U=\Pb(\mathbf{S}=U : \mathbf{S}\subseteq U)$ up to an additive error of at most $m^{-2}$ with high probability. If the estimate is less than $(1-\frac{1}{2m})$ then with high probability $q_U \leq (1-\frac{1}{m})$ otherwise, with high probability we have
\begin{equation}\label{termination}
\mu(U)\leq \sum_{S\in \cC, S\subseteq U} \mu(S) \leq \inparen{1+\frac{4}{m}}\mu(U)
\end{equation}
and the algorithm terminates.

When performing step 2(b) we have a high probability guarantee for the assumption of Lemma~\ref{lemma:largeU} to be satisfied. Hence, we can assume that (by using Lemma~\ref{lemma:largeU} and Lemma~\ref{lemma_additive}) we can find an element $e\in U$ with $p_e=\Pb(e\notin \mathbf{S} : \mathbf{S}\subseteq U)\geq \frac{1}{2m^2}$. Again using Lemma~\ref{lemma_additive} we can perform step 2(c) and obtain a multiplicative $(1+\frac{\eps}{m})$-approximation $\rho_e$ to $p_e$.

Denote the set $U$ at which the algorithm terminated by $U'$ and the elements chosen at various stages of the algorithm by $e_1, e_2, ..., e_l$ with $l=m-|U'|$. The output of the algorithm is:
$$X \defeq  \rho_{e_1}\rho_{e_2}\cdot \cdots \cdot p_{e_l} \mu(U').$$
While the exact value of the sum is
$$Z \defeq p_{e_1} p_{e_2} \cdot \cdots \cdot p_{e_l} \cdot \sum_{S\in \cC, S\subseteq U'} \mu(S).$$
Recall that for every $i=1,2,\ldots, l$ with high probability it holds that:
$$\inparen{1+\frac{\eps}{m}}^{-1}\leq \frac{p_{e_i}}{\rho_{e_i}}\leq \inparen{1+\frac{\eps}{m}}.$$
This, together with \eqref{termination} implies that with high probability:
$$\inparen{1+\frac{\eps}{m}}^{-l} \leq \frac{X}{Z} \leq \inparen{1+\frac{\eps}{m}}^{l} \cdot \inparen{1+\frac{4}{m}},$$
which finally gives $(1+2\eps)^{-1} \leq \frac{X}{Z} \leq (1+2\eps)$ with high probability, as claimed. Note that the algorithm requires $\poly(m, \frac{1}{\eps})$ samples from the oracle in total.
  \end{proof}
  
\bibliographystyle{alpha}
\bibliography{references}

\end{document}